\documentclass[journal]{IEEEtran}
\usepackage{graphicx}

\usepackage[T1]{fontenc}
\usepackage[latin1]{inputenc}
\usepackage{graphicx}
\usepackage{amssymb}
\usepackage{dcolumn}
\usepackage{bm}
\usepackage{color}
\usepackage{multicol}
\usepackage{algorithm}
\usepackage{multirow}
\input epsf
\usepackage{cite,epsfig}
\usepackage{graphicx,amsmath,amssymb,remark,color}

\newtheorem{theorem}{Theorem}
\newtheorem{corollary}{Corollary}
\newtheorem{lemma}{Lemma}
\newtheorem{remark}{Remark}
\newtheorem{definition}{Definition}
\newtheorem{assumption}{Assumption}

\newcommand{\EQ}{\begin{eqnarray}}
\newcommand{\EN}{\end{eqnarray}}
\newcommand{\EQQ}{\begin{eqnarray*}}
\newcommand{\ENN}{\end{eqnarray*}}

\DeclareMathAlphabet{\matheur}{U}{eur}{m}{n}
\DeclareMathAlphabet{\matheurb}{U}{eur}{b}{n}
\DeclareMathAlphabet{\matheus}{U}{eus}{m}{n}
\DeclareMathAlphabet{\matheuf}{U}{euf}{m}{n}
\renewcommand{\t}{^{\mbox{\tiny\sf T}}}

\newcommand{\red}{\color{red}}
\newcommand{\blue}{\color{blue}}    

\newcommand{\bremark}{\begin{remark}
\begin{rm}}
\newcommand{\eremark}{ \end{rm}\hfill \rule{1mm}{2mm}
\end{remark} }
\newcommand{\btheorem}{\begin{theorem} \begin{it}}
\newcommand{\etheorem}{\end{it} \hfill \rule{1mm}{2mm}
\end{theorem} }
\newcommand{\blemma}{\begin{lemma} \begin{it} }
\newcommand{\elemma}{ \end{it} \hfill\rule{1mm}{2mm}
\end{lemma} }
\newcommand{\bcorollary}{\begin{corollary} \begin{it} }
\newcommand{\ecorollary}{ \end{it} \hfill\rule{1mm}{2mm}
\end{corollary} }
\newcommand{\bdefinition}{\begin{definition} }
\newcommand{\edefinition}{  \end{definition} }
\newcommand{\bproposition}{\begin{proposition} }
\newcommand{\eproposition}{\hfill \rule{1mm}{2mm}
\end{proposition} }
\newcommand{\bexample}{\begin{example} \begin{rm}}
\newcommand{\eexample}{ \end{rm} \hfill\rule{1mm}{2mm}
\end{example} }

\newcommand{\bassumption}{\begin{assumption} }
\newcommand{\eassumption}{
\end{assumption} }

\newcommand{\balgorithm}{\medskip\begin{algorithm} \rm}
\newcommand{\ealgorithm}{ \hfill \rule{1mm}{2mm}\medskip
\end{algorithm} }
\newcommand{\barray}{\begin{array} }
\newcommand{\earray}{ \end{array}}

\newcommand{\basm}{\begin{assumption} \begin{rm} }
\newcommand{\easm}{ \end{rm} \hfill\rule{1mm}{2mm}
\end{assumption} }

\ifCLASSINFOpdf

\else

\fi

\renewcommand{\t}{^{\mbox{\tiny\sf T}}}

\begin{document}

\title{Collective Dynamics and Control for Multiple Unmanned Surface Vessels}
\author{ Bin Liu, Zhiyong Chen, \textit{Senior Member}, \textit{IEEE}, Hai-Tao Zhang, \textit{Senior Member}, \textit{IEEE}, Xudong Wang, Tao Geng, Housheng Su, and Jin Zhao, \textit{Senior Member}, \textit{IEEE}
\thanks{This work was supported by the National Natural Science Foundation of China (NNSFC) under Grants with Nos.~ U1713203, 51729501 and 61673189, in part by the Guangdong Innovative and Entrepreneurial
Research Team Program under Grant 2014ZT05G304. (Corresponding author: H.-T. Zhang)

B. Liu, X. Wang and T. Geng are with the Guangdong HUST Industrial Technology Research
Institute, Guangdong Province Key Lab of Digital Manufacturing Equipment, Dongguan 523808, China, emails: binliu92@hust.edu.cn, wangxd2016@hust.edu.cn, dr.geng@aliyun.com.

 Z. Chen is with School of Electrical Engineering and Computing,
The University of Newcastle, Callaghan, NSW 2308, Australia, email: zhiyong.chen@newcastle.edu.au.

H.-T. Zhang is with the School of Automation, the State Key Lab of Digital Manufacturing Equipment
and Technology, and the Key Lab of Imaging Processing and Intelligence
Control, Huazhong University of Science and Technology, Wuhan 430074,
China, email: zht@mail.hust.edu.cn.

H. Su and J. Zhao are with the School of Automation, and the Key Lab of Imaging Processing and Intelligence
Control, Huazhong University of Science and Technology, Wuhan 430074,
China, emails: houshengsu@gmail.com, jinzhao617@hust.edu.cn.
}
}

\maketitle
\begin{abstract}
A multi-unmanned surface vessel (USV) formation control system is established on a novel platform
composed of three 1.2 meter-long hydraulic jet propulsion surface vessels, a differential GPS reference station, and inter-vessel Zigbee communication modules. The system is also equipped with
an upper level collective multi-USV protocol and a lower level vessel dynamics controller.
The system is capable of chasing and surrounding a target vessel. The results are supported by rigorous theoretical analysis  in terms of
asymptotical surrounding behavior and  trajectory regulation.
Extensive experiments are conducted to demonstrate the effectiveness and efficiency of the proposed
hardware and software architectures.

\end{abstract}

\begin{IEEEkeywords}
Multi-agent systems, unmanned surface vessels, collective control, regulation,  underactuated control.
\end{IEEEkeywords}

\IEEEpeerreviewmaketitle

\section{Introduction}\label{introduction}

Unmanned surface vessels (USVs) have extensive applications in marine resource exploration, water pollution clearance, disaster searching and rescue,
marine patrol and prospection, for their low-cost, high efficiency, agility and flexibility.
Most existing research on USVs focuses on a single vessel. As representative works, recurrent neural network-based predictive controllers were designed in \cite{zh12,yan12} to address the nonlinearity of the USV dynamics. Trajectory tracking controllers were proposed for path planning of USVs subject to input saturation, system uncertainties, and wind/wave disturbances in \cite{kd04,kd16,huang15,ch17,so10}.


With the tremendous development over the past years, multi-USV systems have become indispensable tools for
developing marine economic, protecting marine environment, and preserving marine rights.
In particular,  a single USV is far less capable than a multi-USV formation, especially in fulfilling complex tasks of patrol, rescue, smuggle seizing, water pollution clearance, and material delivery.  For example, in harsh marine environments with severe external disturbances, a single USV is more vulnerable than a multi-USV setting where one malfunctioned USV can be replaced and/or rescued by another.

In the field of formation control of multiple unmanned vehicle/robot/vessel, called a multi-agent system (MAS) in general,
these years have witnessed many research outcomes, including $\alpha$-lattice flocking in \cite{olf06}, a second-order Cucker-Smale model in \cite{cu07} and its prediction version in \cite{zhang16,ch15}, homogeneous  and heterogeneous collective circular motion control protocols in \cite{ce08,ch11b,ch13c}, an arbitrary collective closed envelope motion control scheme in \cite{zh16}, and {\blue formation control protocols for Euler-Lagrangian systems in \cite{nu14}.}
More results can be referred to in the survey papers, e.g., \cite{kw15,kn16}.

Especially on formation control of multi-USVs,  the representative works are discussed as follows.
A sliding-mode formation control scheme was designed in \cite{fa07}  for USVs to form arbitrary formations.
A coordinative control protocol governing a multi-USV system was developed in \cite{dong08}
 to a desirable stationary formation with identical orientations.
Formation control of  USVs in the presence of uncertainties and ocean disturbances was studied in \cite{pe13}.
Based on a fuzzy estimator, a distributed constrained control law was proposed in \cite{peng17}
for multiple USVs guided by a virtual leader moving along a parameterized path.
A smooth time-varying distributed control law was proposed in \cite{xie17} that assures that
a multiple USV can globally exponentially converge to a desirable geometric formation.
{\blue The objective of this paper is to drive a team of vessels to surround a target vessel within their convex hull, which is
different from the aforementioned formation control. A relevant theoretical work can found in \cite{ch10}
where the vehicles are initially placed within a
circle and/or using a predefined stand-off distance between the vehicles and the target.
A novel  kinematic control scheme is proposed in this paper that does not require such an initial setup.  }

Most of the aforementioned works focus on formation control protocols of kinematic models,
but not taking complicated surface vessel
dynamics into consideration.
It is of great theoretical challenge to consider the complicated interaction of an upper level collective multi-USV protocol and a lower level vessel dynamics controller. Specifically, this paper answers how to achieve the upper level collective behavior subject to the regulation error from the lower level controller, as well as how to drive the regulation error to zero exponentially for a specified trajectory from the upper level.

Also, theoretical research has rarely been tested in real environment
due to the challenging practical issues in establishment of a real experimental platform.
Rare relevant results can be found in \cite{seo15,guo14,wil17,yu11b} where
experiments were conducted on real water surfaces including rivers, lakes and seas.
These works however focus on a single USV.
In this paper, we aim to test the design in a real lake-based multi-USV formation control platform
that is composed of three 1.2 meter-long jet-propelling vessels equipped with onboard differential GPS receivers and imaging processing modules,
located at Songshan Lake, Guangdong, China.

\section{Modeling}
\label{se: modeling}


Consider a multi-USV system consisting of $N\geq 3$ vessels. Let
$\mathbb N=\{1,2,\cdots, N\}$. Denote
the complete position distribution of the system as
$x=[x_1\t,\dots, x_N\t]\t$, where $x_i=[x_{1,i},x_{2,i}]\t\in \mathbb R^2$, $i\in \mathbb N$,
represents the Cartesian coordinates of  the $i$-th vessel.
Denote ${\rm co}(x)$ be the convex hull of $x_1$, $\cdots$, $x_N$, that is,
$${\rm co}(x):= \left\{ \sum_{i=1}^{N} \lambda_i x_i:\; \lambda_i \geq 0, \forall i \;{ \rm and }
\sum_{i=1}^N \lambda_i = 1 \right\}.$$
Also, let
\EQ P_{x_o} (x) :=\min_{s \in {\rm co}(x) } \|x_o - s \|\label{eq: Pxo}\EN
be the distance between a point $x_o$  and ${\rm co}(x)$. Obviously,
$x_o\in{\rm co}(x)$ if and only if $P_{x_o} (x)  =0$.

The kinematics model for each vessel is given as follow
\EQ
\dot x_{i}=S(\psi_i) \left[\begin{array}{c} w_i  \\
v_i \end{array}\right] , \;i \in {\mathbb N}
 \label{eq: multi-USV}
\EN
for a rotation matrix
\EQ
S(a) := \left[\begin{array}{cc} \cos a & -\sin a  \\
\sin a & \cos a
\end{array}\right], \; a\in \mathbb R.
\EN
In the model,  $w_i$, $\psi_i$ and  $v_i$ represent the forward (surge) velocity, the orientational (yaw) angle and the transverse (sway) velocity, respectively, as illustrated in Fig.~\ref{fig: boatdirection}.
Denote the orientational (yaw) angular velocity $r_i, i \in {\mathbb N}$.

\begin{figure}[htbp]
\centering
\includegraphics[width=8cm]{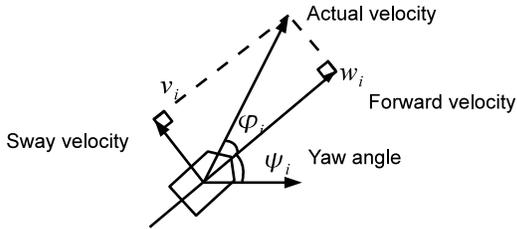}
\caption{Illustration of kinematics model of a vessel.}
\label{fig: boatdirection}
\end{figure}

A complete but complicated nonlinear dynamics model has been proposed in literature based on physical
principles with the simplified hydrodynamic effects; see, e.g., Eqs.~(4-6) of
\cite{ch13e}.  The model was identified with the nominal forward speed up to  20~knots.
Therein, several simplified linear variants of the dynamic equations and some
control design approaches are also discussed for trajectory tracking including cascaded PD and backstepping control.
With the same objective ``to obtain a model that is rich enough to enable effective motion planning and control, simple enough for experimental identification, and general enough to describe a variety of vehicles operating over a large range of speeds,'' we use the following equations for the dynamics of vessels used in the paper, for $i \in {\mathbb N}$,
\EQ
\label{dmodel}
\barray{rll}
\dot\psi_i &=& r_i, \\
\dot{w_i}&=&k_1w_i+k_2v_ir_i+k_3\tau_{i,1},\\
\dot{r_i}&=&k_4r_i+k_5\tau_{i,2},\\
\blue \dot{v_i}&\blue =&\blue k_6v_i+k_7w_ir_i,
\earray
\EN
where the two control variables are the propeller speed $\tau_{i,1}$ and the steering angle $\tau_{i,2}$.
Denote $\tau_i = [\tau_{i,1}, \tau_{i,2}]\t$.
In particular, this model is given for the vessels working in a medium speed mode  (1-3m/s).
In this model, we ignore the high-order nonlinearities except the cross nonlinearity $v_ir_i$ {\blue and $w_ir_i$} in the {\blue second and fourth} equations of (\ref{dmodel}). This simplification is based on extensive experiments and data matching.

\section{Problem Formulation and Controller Framework}\label{sec: problem}

The main technical challenge in multi-USV system control is
to propose a decentralized protocol that achieves the specified collaborative behavior
through the control to each vessel's dynamics model.
Some preliminary manipulation is first introduced in this section.

Let $w_i^r$, $v_i^r$, and $\psi_i^r$ be the desired signals for $w_i$, $v_i$, and
$\psi_i$, respectively, in the kinematics model  (\ref{eq: multi-USV}).
Denote
\EQ
\barray{rll}
\widetilde w_i &:=& w_i - w_i^r, \\
\widetilde v_i  &:=& v_i - v_i^r, \\
\widetilde \psi_i  &:=& \psi_i - \psi_i^r.
\earray
\EN
Direct calculation shows that
\EQ \label{kmodel}
\dot x_{i}
 = S( \psi_i^r )\left[\begin{array}{c} w_i^r
\\
 v_i^r   \end{array}\right]
 + e_i
\EN
with
\EQ
e_i = [S( \psi_i^r +\widetilde \psi_i ) -S( \psi_i^r )]  \left[\begin{array}{c} w_i^r
\\
 v_i^r   \end{array}\right]
 +S(\psi)\left[\begin{array}{c}  \widetilde w_i
\\
 \widetilde v_i  \end{array}\right]. \label{ei}
 \EN

The control design framework of this paper consists of the following two steps.
\begin{itemize}

\item[(i)] {\it (\textbf{Upper level collective protocol}) To design the desired $w_i^r$, $v_i^r$, and $\psi_i^r$, for the
kinematics model (\ref{kmodel}) such that the multi-USV achieves
a desired collective behavior, subject to the perturbation $e_i(t)$ approaching zero.}
\item[(ii)] {\it (\textbf{Lower level vessel dynamics control}) To design the actuator input $\tau_i$ for the dynamics model (\ref{dmodel})
such that $w_i$, $v_i$, and $\psi_i$ achieve
the desired $w_i^r$, $v_i^r$, and $\psi_i^r$ given in (i), in particular, with
 $e_i(t)$ approaching zero.}

\end{itemize}

The technical objective of this paper is to propose solutions to the two steps.
A direct conclusion is as follows, with the two steps in the aforementioned
framework solved, the closed-loop system composed of   (\ref{dmodel}),
(\ref{kmodel}) and the actuator input $\tau_i$,
achieves the desired collective behavior specified in step (i).
To be more specific, two collective behaviors, i.e., surrounding control and equally surrounding control, are studied in this paper,
 The rigorous definitions are given below.

%

\bdefinition  A target vessel position $x_o\in\mathbb R^2$ is {\it asymptotically surrounded} by
the $N$ vessels of the complete position distribution $x$ if
\EQ
\label{eq: chasemodel}
\barray{rll}
\lim_{t\rightarrow\infty} P_{x_o}(x(t)) =0.
\earray
\EN
\edefinition

\bdefinition
 A target vessel position $x_o\in\mathbb R^2$ is {\it asymptotically equally surrounded} by
the $N$ vessels of the complete position distribution $x$ if
it is asymptotically surrounded by them with
 \EQ
\label{eq: chasemodel1}
\barray{rll}
&&\lim_{t\rightarrow\infty} \|x_i(t) -x_o\| =\rho_o, \; i \in {\mathbb N}\\
&&\lim_{t\rightarrow\infty}\|x_i(t)-x_j(t)\| \geq d, \; i \neq j \in {\mathbb N}
\earray
\EN
for $\rho_o > 0$ and $d=2\rho_o\sin(\pi/N)$.
\edefinition

\begin{remark}
For every two adjacent vessels, say $\hbar$ and $\ell$,
the property (\ref{eq: chasemodel1}), together with the geometric constraints,
implies
$\lim_{t\rightarrow\infty}\|x_\hbar(t)-x_\ell(t)\| = d$.

\end{remark}

We propose two approaches in Cartesian coordinate
and polar coordinate, respectively, to achieve different collective
behaviors with different features.

\medskip

{\bf Approach 1:} For any signal $u_i^r \in {\mathbb R}^2$ to be designed {\blue and an arbitrary
$v_i^r <\|u_i^r\|$,}  let\footnote{For a vector $x=[x_1, x_2]\t  \in {\mathbb R}^2$, let $\angle x \in [0, 2\pi)$ be
the angle of the complex number $x_1 + i x_2$ in the complex plan.
}
{ \blue
 \EQ
\barray{rll}
 w_i^r &=& \sqrt{\|u_i^r\|^2-(v_i^r)^2},  \\
 \psi_i^r &=&  2\kappa \pi + \angle u_i^r -\mbox{atan} ({v_i^r}/{w_i^r}),  \\
  \earray \label{conv-a}
 \EN
 }where {\blue $\mbox{atan} ({v_i^r}/{w_i^r})$ is the drift angle,} $\kappa$ is an integer-valued signal such that $\kappa(0)=0$, i.e.,
$\psi_i^r(0)=  \angle u_i^r (0)$ and a continuous $\angle u_i^r(t)$
implies a continuous $\psi_i^r(t)$ in time $t$.

Accordingly, one has that
 {\blue
 \EQQ
u^r_i =
\left[\begin{array}{c} \cos \angle u_i^r
\\
 \sin \angle u_i^r   \end{array}\right]   \|u_i^r\|  =
 S( \psi_i^r )\left[\begin{array}{c} w_i^r
\\
 v_i^r   \end{array}\right].
 \ENN
 }
 Then, the model (\ref{kmodel}) becomes
 \EQ   \label{kmodel-a}
 \dot x_i = u_i^r+e_i.
 \EN

Obviously, the aforementioned step (i) is solvable with
{\blue an arbitrary  $v_i^r <\|u_i^r\|$} and the desired $w_i^r$ and $\psi_i^r$
 given by (\ref{conv-a}) if the following step is solvable.

 \begin{itemize}

\item[(i$'$)] {\it To design a desired $u_i^r$ for the
kinematics model (\ref{kmodel-a}) such that the multi-USV achieves
asymptotically surrounding formation, subject to the perturbation $e_i(t)$ approaching zero.}

\end{itemize}

\medskip

{\bf Approach 2:}  For a specified target vessel position $x_o\in\mathbb R^2$,
let
  \EQQ
  \barray{rll}
  \rho_i &:=& \|x_i -x_o\|,\\
 \theta_i &:=& 2 \kappa \pi + \angle (x_i-x_o)
 \earray
 \ENN
be the polar coordinate of the $i$-th vessel
where $\kappa$ is an integer-valued signal such that $\kappa(0)=0$, i.e.,
$\theta_i(0)=  \angle x_i (0)$ and a continuous $\angle x_i(t)$
implies a continuous $\theta_i(t)$ in time $t$.

For any signals $\eta_i^r,  \omega_i^r \in {\mathbb R}$, let
\EQ \label{conv-b}
u_i^r = S(\theta_i)\left[\begin{array}{c} \eta_i^r
\\
\rho_i \omega_i^r  \end{array}\right]  \EN
and hence  $w_i^r$, $v_i^r$, and $\psi_i^r$
 given in Approach 1.
Define
\EQQ  \left[\begin{array}{c} \widetilde\eta_i^r
\\
 \widetilde\omega_i^r  \end{array}\right]:=
\left[ \barray{cc}1& 0 \\
 0 & 1/\rho_i \earray \right]S^{-1}(\theta_i)  e_i,
\ENN
i.e.,
 \EQQ e_i = S(\theta_i) \left[\begin{array}{c} \widetilde\eta_i^r
\\
\rho_i \widetilde\omega_i^r  \end{array}\right].
\ENN
Note the following calculation
 \EQQ \dot x_i = u_i^r+e_i =S(\theta_i)\left[\begin{array}{c} \eta_i^r
\\
\rho_i \omega_i^r  \end{array}\right]+S(\theta_i) \left[\begin{array}{c} \widetilde\eta_i^r
\\
\rho_i \widetilde\omega_i^r  \end{array}\right] \\
=S(\theta_i)\left[\begin{array}{c} \eta_i^r   +\widetilde\eta_i^r
\\
\rho_i (\omega_i^r +\widetilde \omega_i ^r)\end{array}\right], \ENN
and
\EQQ \dot x_i = \dot \rho_i \left[\barray{c}\cos \theta_i \\
\sin \theta_i\earray\right] +\rho_i
\left[\barray{c} -\sin \theta_i \\
\cos \theta_i\earray\right]\dot\theta_i
=S(\theta_i) \left[\begin{array}{c} \dot\rho_i \\
\rho_i \theta_i \end{array}\right].
\ENN
Then, the model (\ref{kmodel}) becomes
\EQ \label{kmodel-b}
 \dot \rho_i =\eta_i^r   +\widetilde\eta_i^r,\;
\dot\theta_i = \omega_i^r +\widetilde \omega_i^r.
\EN
Also, it is noted that ${\blue \widetilde\eta^r_i(t)}$ approaches zero if
$e_i(t)$ approaches zero;
 ${\blue \widetilde\omega^r_i(t)}$ approaches zero if
$e_i(t)$ approaches zero and $\rho_i(t)$ is asymptotically lower bounded by a positive constant.

Obviously, the aforementioned step (i) is solvable with the desired $w_i^r$, $v_i^r$, and $\psi_i^r$
 given by (\ref{conv-a}) and (\ref{conv-b})  if the following step is solvable.

 \begin{itemize}

\item[(i$''$)] {\it To design desired $\eta_i^r$ and $\omega_i^r$ for the
kinematics model (\ref{kmodel-b}) such that
$\rho_i(t)$ is asymptotically lower bounded by a positive constant
subject to the perturbation ${\blue \widetilde\eta^r_i(t)}$  approaching zero; and
the multi-USV achieves
equally asymptotically surrounding formation subject to the perturbation ${\blue \widetilde\eta^r_i(t)}$ and ${\blue \widetilde\omega^r_i(t)}$
 approaching zero.}

\end{itemize}

In what follows, we aim to  propose solutions to the steps (i$'$) and (i$''$) in Section~\ref{sec:collective},
and afterwards the step (ii) in Section~\ref{sec:reg}.



\section{Collective Control Design}
\label{sec:collective}

This section aims to propose a controller for each vessel so that the multi-USV achieves
the desired asymptotically surrounding formation in the sense given in (i$'$) or (i$''$).

\subsection{Asymptotically  Surrounding Control}

The main objective of this subsection is described in step (i$'$). More specifically, it aims
to design the desired $u_i^r$ for the
kinematics model (\ref{kmodel-a}) such that a specified target vessel position $x_o$ (may be an enemy vessel)
is asymptotically surrounded by the USV team, subject to the perturbation $e_i(t)$ approaching zero.

To give the desired $u_i^r$ in a distributed manner, we define the set of neighbors of vessel $i$
as $$\mathcal N_i :=\left\{ j\in {\mathbb N} : j\neq i,  \; |\|x_i-x_j\|<\mu \right\},\;  i\in {\mathbb N} $$ with a specified distance $\mu >0$.
First, assume that $x_o$ is available for all vessels,  then the control law for each follower is designed as follows,
with $x_{ij} :=x_i -x_j$ and $x_{oi}:={\blue x_o}-x_i$ throughout the paper,
\EQ
\label{eq: controlui}
 u_i^r=  \gamma_1 \sum_{j\in \mathcal N_i} (\mu^2 -\|x_{ij}\|^2) x_{ij}
 +\gamma_2 x_{oi}.
\EN

 Now, the main technical result is stated in the following theorem.

\begin{theorem}
\label{theo: surround1}
For the system (\ref{kmodel-a}) with $\lim_{t\rightarrow\infty}e_i(t)=0$ exponentially
and the controller  (\ref{eq: controlui}) with $\gamma_1, \gamma_2 >0$,
the states of the closed-loop system are bounded. Moreover, the target vessel position $x_o$ is  asymptotically surrounded by
the $N$ vessels in the sense of (\ref{eq: chasemodel}).
 \end{theorem}

\begin{proof} The closed-loop system composed of
 (\ref{kmodel-a}) and (\ref{eq: controlui})
can be put in the following form
 \EQ   \label{kmodel-a-cl}\barray{rll}
  \dot x_i &=&u_i^r+e_i\\
 &=&\gamma_1 \sum_{j\in \mathcal N_i} (\mu^2 -\|x_{ij}\|^2) x_{ij}
 +\gamma_2 x_{oi}  +e_i.\earray
 \EN

Let
\EQQ
V_o (x_{ij})=\left\{ \begin{array}{ll} (\|x_{ij}\|^2-\mu^2)^2 & \|x_{ij}\| < \mu \\
0 & \|x_{ij}\|\geq\mu\end{array}
\right.
\ENN that is continuously differentiable and whose
derivative is 0 for  $\|x_{ij}\| \geq \mu$ and
\EQQ
\dot {V}_o(x_{ij})=\frac{\partial V_o(x_{ij})}{\partial x_{ij}} \left[ \frac{\partial x_{ij}}{\partial x_i}\dot x_i
+ \frac{\partial x_{ij}}{\partial x_j}\dot x_j  \right]\\
=4(\|x_{ij}\|^2-\mu^2)   x_{ij}\t \dot x_{ij}
\ENN
 for $ \|x_{ij}\| < \mu $.
Let
\EQQ
V_1(x) =\frac{\gamma_1}{4} \sum_{i, j\in {\mathbb N}, j\neq i}V_o(x_{ij})
\ENN
whose derivative is, due to the symmetric property of the undirected graph,
\EQQ
\dot V_1(x)
=2\gamma_1 \sum_{i \in {\mathbb N}, j\in \mathcal N_i} (\|x_{ij}\|^2-\mu^2)   x_{ij}\t \dot x_{i} .
\ENN
Let
\EQQ V_2(x) = \gamma_2\sum_{i\in \mathbb N}\|x_{oi}\|^2.\ENN
Analogously, one has
\EQQ \dot V_2(x) = 2\gamma_2\sum_{i\in \mathbb N}x_{oi}\t \dot x_i.
\ENN
The derivative of $V(x) = V_1(x)+V_2(x)$, along the trajectory of (\ref{kmodel-a-cl}), is
\EQQ
\dot V(x) = 2 \sum_{i\in \mathbb N} \left[ \gamma_1 \sum_{ j\in \mathcal N_i} (\|x_{ij}\|^2-\mu^2)   x_{ij}\t
+2\gamma_2 x_{oi}\t   \right] \t \dot x_{i}\\
=- 2 \sum_{i\in \mathbb N} (u_i^r)\t (u_i^r +e_i)
\leq  -\sum_{i\in \mathbb N} \|u_i^r\|^2  + \sum_{i\in \mathbb N}\|e_i\|^2.
\ENN
Denote
\EQQ
U(t) = -\sum_{i\in \mathbb N}  \int_0^{t} \|u_i^r(s)\|^2 ds \leq 0.
\ENN
Direct calculation gives
\EQQ 0\leq V(x(t))  \leq U(t)  + \sum_{i\in \mathbb N} \int_{0}^t \|e_i(s)\|^2 ds + V(x(0)).
\ENN
As a result, $V(x)$ is upper bounded, so is the state $\|x(t)\|$.


To prove the moreover part, let $\bar x =\sum_{i=1}^{N} x_i /N$, and $\bar e =\sum_{i=1}^{N} e_i /N$.
 Then,
 \EQQ
 \dot {\bar x} &=&\frac{ \gamma_1}{N} \sum_{i=1}^{N}
  \sum_{j\in \mathcal N_i} (\mu^2 -\|x_{ij}\|^2) x_{ij}
 \\&&+\frac{ \gamma_2}{N} \sum_{i=1}^{N}   (x_{oi})   +
 \frac{1}{N} \sum_{i=1}^{N} e_i
  \\&=&  - \gamma_2 \bar x  + \gamma_2 x_o + \bar e,
 \ENN
that implies
$\lim_{t\rightarrow\infty}\bar x(t) - x_o =0$
and hence  (\ref{eq: chasemodel}). The proof is thus completed.
 \end{proof}

\medskip

 Next, we will investigate the decentralized scenario that $x_o$ is not available for all the vessels. In such a situation, a decentralized estimator is required for each follower vessel to estimate $x_o$.
Define ${\mathbb N}_{1}$ as the set of vessels that can detect the target (i.e., leaders) and ${\mathbb N}_{2}$ as the set of vessels that cannot (i.e., followers). Let $\mathcal M_i$ be the set of communication neighbors of vessel $i$, $\forall i\in \mathbb N$.
  The estimator is designed as follows:
\begin{equation}
\label{eq: estimator}
\dot{y}_i=
\left\{
\begin{array}{ll}
\gamma_3[\sum_{j\in \mathcal M_i}(y_j-y_i)+({\blue x_o}-y_i)],& i\in {\mathbb N}_{1}\\
\gamma_3\sum_{j\in \mathcal M_i}(y_j-y_i),& i\in {\mathbb N}_{2}
\end{array}
\right.
\end{equation}
where $y_i$ is the estimate of $x_o$ for vessel $i$.
As a result, the controller (\ref{eq: controlui}) is modified as follows
 \EQ
\label{eq: controlui_est}
 u_i^r=  \gamma_1 \sum_{j\in \mathcal N_i} (\mu^2 -\|x_{ij}\|^2) x_{ij}
 +\gamma_2 (y_i-x_i).
\EN
A similar statement still holds as in the following theorem.

\begin{theorem}
For the system (\ref{kmodel-a}) with $\lim_{t\rightarrow\infty}e_i(t)=0$ exponentially
and the controller (\ref{eq: estimator}) and (\ref{eq: controlui_est}) with $\gamma_1, \gamma_2,\gamma_3 >0$,
the states of the closed-loop system are bounded if the network determined by $\mathcal M_i$, $\forall i\in \mathbb N$,
is connected and the target can be detected by at least one vessel (i.e., ${\mathbb N}_{1}\neq \emptyset$).
Moreover, the target vessel position $x_o$ is  asymptotically surrounded by
the $N$ vessels in the sense of (\ref{eq: chasemodel}).
\end{theorem}

\begin{proof}
The network (\ref{eq: estimator}) is able to achieve $\lim_{t\rightarrow\infty}\|y_i-x_o\|=0$, $\forall i\in \mathbb N$, exponentially,
if the network determined by $\mathcal M_i$, $\forall i\in \mathbb N$,
is connected and the target can be detected by at least one vessel.
 Let $\epsilon_i = y_i-x_o$. One has $\lim_{t\rightarrow\infty} \epsilon_i(t)=0$, exponentially.

The closed-loop system composed of
(\ref{kmodel-a})   and (\ref{eq: controlui_est})
can be rewritten as
 \EQ   \label{kmodel-a-cl-est}
 \dot x_i =\gamma_1 \sum_{j\in \mathcal N_i} (\mu^2 -\|x_{ij}\|^2) x_{ij}
 +\gamma_2 x_{oi} +\bar e_i
\EN
 for $\bar e_i=  \gamma_2 \epsilon_i +e_i$. Clearly,
  $\lim_{t\rightarrow\infty}\bar e_i(t)=0$ exponentially.
So, the system (\ref{kmodel-a-cl-est}) takes the same form
of (\ref{kmodel-a-cl}). The proof follows that of Theorem~\ref{theo: surround1}.
\end{proof}


\subsection{Asymptotically Equally Surrounding Control}

The main objective of this subsection is described in steps (i$''$). More specifically, it aims
to design the desired  $\eta_i^r$ and $\omega_i^r$ for the
kinematics model (\ref{kmodel-b}) such that a specified target vessel position $x_o$ is
asymptotically equally surrounded by the multi-USV, subject to the perturbation ${\blue \widetilde\eta^r_i(t)}$ and ${\blue \widetilde\omega^r_i(t)}$
 approaching zero.
In (\ref{kmodel-b}),   $\rho_i$ and $\theta_i$ are the radius and the moving angle of vessel $i$, respectively, relative to the
target $x_o$.
Define
$\theta_{ij} :=\theta_i -\theta_j +2\kappa\pi \in [-\pi, \pi)$ where
$\kappa$ is an integer-valued signal.

 The main purpose of $\eta_i^r$ is to drive all vessels to a circle of a specified radius $\rho>0$.
 The controller for $\eta^r_i$ takes the following linear structure
\EQ
\label{controleta}
\eta^r_i:=\beta_1(\rho_o-\rho_i).
\EN
To put the desired $\omega_i^r$ in a distributed manner, we define the set of neighbors of vessel $i$
as
\EQQ \Theta_i =\left\{ j\in {\mathbb N} : j\neq i,  \; |\theta_{ij}| < \frac{2\pi}{N} \right\},\;  i\in {\mathbb N}. \ENN
Then, $\omega_i^r$ is designed such that the angles of the vessels change along the negative gradient of
an energy function $P_o(\theta_{ij})$ to be specified later, that is,
\begin{equation}
\label{eq: controlomega}
\omega_i^r=
\beta_2 \sum_{j\in \Theta_i} \left(\frac{2\pi}{N}-|\theta_{ij}|\right) \frac{\theta_{ij}}{|\theta_{ij}|}.
\end{equation}

\begin{theorem} \label{the: theta}
Consider the system (\ref{kmodel-b})
and the controller (\ref{controleta}) and (\ref{eq: controlomega}) with $\beta_1, \beta_2 >0$.
Then,
\EQ \label{limrho}
\lim_{t\rightarrow\infty} \rho_i(t) =\rho_o, \; i \in {\mathbb N}\EN
if $\lim_{t\rightarrow\infty}{\blue \widetilde\eta^r_i(t)}=0$ exponentially. Moreover,
 the target vessel position $x_o$ is asymptotically equally surrounded by
the $N$ vessels in the sense of (\ref{eq: chasemodel1}), or equivalently, (\ref{limrho}). Furthermore,
 \EQ
\lim_{t\rightarrow\infty} |\theta_{ij}(t)| \geq 2\pi/N , \; i \neq j \in {\mathbb N}, \label{limtheta}
\EN
if $\lim_{t\rightarrow\infty}{\blue \widetilde\omega_i(t)}=0$ exponentially.
 \end{theorem}

 \begin{proof}  The closed-loop system can be rewritten as
 \EQ \label{kmodel-b-cl}
\barray{rll}
\dot \rho_i &=&\beta_1(\rho_o-\rho_i)  +{\blue \widetilde\eta^r_i},\\
\dot\theta_i &=& \beta_2 \sum_{j\in \Theta_i} (\frac{2\pi}{N}-|\theta_{ij}|) \frac{\theta_{ij}}{|\theta_{ij}|}+{\blue \widetilde \omega^r_i}.
\earray
\EN
The proof of (\ref{limrho}) is straightforward from the linear system property.
To prove the phase distribution property (\ref{limtheta}), we define a potential function $P_o(\theta_{ij})$ as follows
\begin{equation}
\label{eq: potentialrij}
P_o(\theta_{ij}):=
\left\{
\begin{array}{lll}
\frac{\beta_2}{2}(|\theta_{ij}|-\frac{2\pi}{N})^2,&|\theta_{ij}| < \frac{2\pi}{N}\\
0,&|\theta_{ij}|\geq \frac{2\pi}{N}
\end{array}
\right.
\end{equation}
that is continuously differentiable and whose
derivative is 0 for  $|\theta_{ij}|\geq \frac{2\pi}{N}$ and

\EQQ
\dot P_o(\theta_{ij}) 
&=&\frac{\partial P_o(\theta_{ij})}{\partial \theta_{ij}}
\left[\frac{\partial \theta_{ij}}{\partial \theta_i} \dot\theta_i +
\frac{\partial \theta_{ij}}{\partial \theta_j} \dot\theta_j  \right] \\
&=&\beta_2 \left(|\theta_{ij}|-\frac{2\pi}{N}\right)\frac{\theta_{ij}}{|\theta_{ij}|}\dot{\theta}_{ij}
\ENN
 for $|\theta_{ij}| < \frac{2\pi}{N}$.
 Let
 \EQQ
P(\theta) =  \sum_{i, j\in {\mathbb N}, j\neq i}P_o(\theta_{ij})
\ENN
whose derivative is, due to symmetric property of the undirected graph,
\EQQ\barray{rll}
\dot P(\theta) &=&\sum_{i, j\in {\mathbb N}, j\neq i}\beta_2 (|\theta_{ij}|-\frac{2\pi}{N})\frac{\theta_{ij}}{|\theta_{ij}|}\dot{\theta}_{ij} \\
&=&2\sum_{i\in {\mathbb N}, j\in \Theta_i}\beta_2 (|\theta_{ij}|-\frac{2\pi}{N})\frac{\theta_{ij}}{|\theta_{ij}|}\dot{\theta}_{i}\\
&=&-2\sum_{i\in {\mathbb N}} \omega^r_i \dot{\theta}_{i}
\\&=&-2\sum_{i\in {\mathbb N}} \omega^r_i (\beta_2 \sum_{j\in \Theta_i} (\frac{2\pi}{N}-|\theta_{ij}|) \frac{\theta_{ij}}{|\theta_{ij}|}+\widetilde \omega_i)
\\ &\leq&  -\sum_{i\in \mathbb N} \|\omega_i^r\|^2  +\sum_{i\in \mathbb N} \|\widetilde \omega_i\|^2.\earray
\ENN
Denote
\EQQ
\Omega(t) := -\sum_{i\in \mathbb N}  \int_0^{t} \|\omega_i^r(s)\|^2 ds \leq 0.
\ENN
Direct calculation gives
\EQQ 0\leq P(\theta(t))  \leq \Omega(t)  + \sum_{i\in \mathbb N} \int\textit{}_{0}^t \|\widetilde\omega_i(s)\|^2 ds + P(\theta(0)).
\ENN
It is noted that  $\Omega(t) $ is lower bounded and monotonic, so $\Omega(t)$ has a finite limit as $t\rightarrow\infty$.
Together with the fact that
$\ddot \Omega(t)$ is bounded, it implies
$\lim_{t\rightarrow\infty}\dot \Omega(t)=0$ and hence
$\lim_{t\rightarrow\infty}\omega_i^r (t)=0$,
by Barbalat's lemma \cite{kh02}.
{\blue From (\ref{eq: controlomega}), one has either $\lim_{t\rightarrow\infty}|\theta_{ij}(t)| = 2\pi/N$
for $j\in \Theta_i$ or $\lim_{t\rightarrow\infty}|\theta_{ij}(t)| \geq 2\pi/N$ for $j\notin \Theta_i$ and $j\neq i$.}
The  property (\ref{limtheta}) thus holds, and the proof is thus completed.
\end{proof}

{\blue\bremark
In the controller (\ref{eq: controlui}) or (\ref{eq: controlui_est}),  the term $(\mu^2 -\|x_{ij}\|^2) x_{ij}$ gives the
repulsive velocity between two vessels.
In (\ref{eq: controlomega}), the term $\left({2\pi}/{N}-|\theta_{ij}|\right) {\theta_{ij}}/{|\theta_{ij}|}$ gives the
repulsive angular velocity between two vessels.
In particular, the closer are the two vessels, the larger is the repulsive velocity.
It provides a mechanism for collision avoidance among the follower vessels.
However, rigorous collision avoidance analysis is an interesting topic for future research.  \eremark}

\section{Trajectory Regulation}
\label{sec:reg}

In this section, we will solve the problem formulated in step~(ii), that is,
to design the actuator input $\tau_i$ for the dynamics model (\ref{dmodel})
such that $w_i$, $v_i$, and $\psi_i$ achieve
the desired $w_i^r$, $v_i^r$, and $\psi_i^r$, respectively.
{\blue Note that (\ref{dmodel}) is an under-actuated system.
The states $w_i$ and $\psi_i$ can be controlled through $\tau_i$
to achieve the desired $w_i^r$ and $\psi_i^r$ as elaborated in Theorem~\ref{thm-regu},
while $v_i$ cannot be directly controlled. Fortunately, the desired
$v_i^r$ can be arbitrarily selected as explained in Section~III
as long as $v_i^r <\|u_i^r\|$. Therefore, we can trivially set $v_i^r =v_i $
that automatically includes regulation of $v_i$ to $v_i^r$.
In the scenario investigated in this paper,  the sway velocity $v_i$ is typically small,
which makes $v_i^r <\|u_i^r\|$ hold in general.  In practice, if
a large $v_i$ occurs in a rare situation, the vessel can be intervened
to reduce its sway velocity.
}

\begin{theorem} \label{thm-regu} For sufficiently smooth desired signals  $w_i^r(t)$ and $\psi_i^r(t)$,
pick a  sufficiently smooth signal
\EQ \varpi_i(t) \geq \max\{1, |w_i^r(t) |\}.\label{eq: varpi}\EN
Define a lumped reference signal
\EQQ
\zeta_i :=[\varpi_i, \dot \varpi_i,\ddot \varpi_i, \psi_i^r, \dot \psi_i^r, \ddot \psi_i^r]\t.
\ENN
For the system (\ref{dmodel}), consider the actuator input $\tau_{i,1}$ as follows
 \EQ
 \barray{rll}
\dot \eta_i &=& -\kappa_1 \widetilde w_i + \ddot w_i^r,\\
\tau_{i,1} &=&(- k_1w_i - k_2v_ir_i  +\eta_i -\kappa_2 \widetilde w_i) / k_3, \label{tau1}
\earray
\EN
for $\widetilde w_i := w_i - w_i^r $
and some positive control parameters $\kappa_1$ and $\kappa_2$
satisfying  $\kappa_1 > \frac{1}{4}\kappa_2^2$;
consider the actuator input $\tau_{i,2}$ as follows
 \EQ
  \barray{rll}
\tau_{i,2} &=& g(r_i,  \psi_i,\zeta_i)\\
&=&
[ k_4 r_i \varpi_i
+2 r_i \dot  \varpi_i
 - \ddot \psi_i^r \varpi_i
 -2 \dot \psi_i^r \dot\varpi_i
  + \widetilde \psi_i  \ddot\varpi_i   \\
&& -\kappa_3^2 \widetilde \psi_i \varpi_i  + (\kappa_3
+\kappa_4)  \widetilde r_i ]/(-k_5 \varpi_i )
 \earray \label{tau2}
\EN
for
\EQQ \widetilde \psi_i  &=& \psi_i - \psi_i^r,\\
\widetilde r_i   &=& r_i \varpi_i  - \dot \psi_i^r \varpi_i + \widetilde \psi_i \dot\varpi_i + \kappa_3 \widetilde \psi_i \varpi_i
\ENN
and some positive control parameters $\kappa_3$ and $\kappa_4$.
Then,  the states $w_i$ and $\psi_i$ achieve
the desired $w_i^r$ and $\psi_i^r$,  in particular, with
 $e_i(t)$ in (\ref{ei}) approaching zero exponentially if $v_i^r =v_i$ is bounded.
 \end{theorem}

\begin{proof}
The $w_i$-dynamics and the controller (\ref{tau1}) can be put in the following
form, with $\widetilde w_i = w_i - w_i^r$ and $\widetilde \eta_i =\eta_i -\dot w_i^r$,
 \EQ
\label{tildaeta}
\barray{rll}
\dot{ \widetilde \eta} &=& -\kappa_1 \widetilde w_i , \\
\dot{\widetilde w}_i &=&-\kappa_2 \widetilde w_i +\widetilde\eta_i,
\earray
\EN
which is exponentially stable when $ \kappa_1 > \frac{1}{4}\kappa_2^2$.

For $\widetilde \psi_i  = \psi_i - \psi_i^r $, one has
\EQQ \dot {\widetilde \psi}_i = r_i - \dot \psi_i^r.
\ENN
Let $\varphi_i = \widetilde \psi_i \varpi_i$. Then,
\EQ  \label{tildevarphi}
\barray{rll}
\dot \varphi_i &=& (r_i - \dot \psi_i^r) \varpi_i + \widetilde \psi_i \dot\varpi_i,  \\
&=& -\kappa_3 \varphi_i  +\widetilde r_i
\earray
\EN where
\EQQ
\widetilde r_i
= r_i \varpi_i  - \dot \psi_i^r \varpi_i + \widetilde \psi_i \dot\varpi_i + \kappa_3 \widetilde \psi_i \varpi_i\\
 = (r_i - \dot \psi_i^r) \varpi_i + \widetilde \psi_i \dot\varpi_i + \kappa_3 \varphi_i .
\ENN
Direct calculation gives
\EQQ
\dot {\widetilde r}_i
  =k_4 r_i \varpi_i  +k_5 \tau_{i,2} \varpi_i
+r_i \dot  \varpi_i
 - \ddot \psi_i^r \varpi_i
 - \dot \psi_i^r \dot\varpi_i \\
 +\dot { \widetilde \psi}_i \dot\varpi_i + \widetilde \psi_i \ddot\varpi_i
 + \kappa_3 \dot\varphi_i .\ENN
Noting that
\EQQ
\tau_{i,2}   = g(r_i,  \psi_i,\zeta_i)
=[ k_4 r_i \varpi_i
+r_i \dot  \varpi_i
 - \ddot \psi_i^r \varpi_i
 - \dot \psi_i^r \dot\varpi_i \\
 +\dot { \widetilde \psi}_i \dot\varpi_i + \widetilde \psi_i \ddot\varpi_i
 + \kappa_3 \dot\varphi_i
+\kappa_4  \widetilde r_i ] /(-k_5 \varpi_i ),
 \ENN
one has
\EQ \label{tilder}
\dot {\widetilde r}_i
 = -\kappa_4  \widetilde r_i.  \EN

From (\ref{tildaeta}), (\ref{tildevarphi}), and (\ref{tilder}),  one has
\EQQ
\lim_{t\rightarrow\infty} \widetilde w_i(t) =0,\;
\lim_{t\rightarrow\infty}   \varphi_i(t) =0,\;
\lim_{t\rightarrow\infty} \widetilde   r_i(t) =0
\ENN
exponentially. {\blue Furthermore, it follows from (\ref{eq: varpi}) and (\ref{tildevarphi}) that $\lim_{t\rightarrow\infty} \widetilde \psi_i(t) =0$ exponentially.}
{\blue
It can be verified that
\EQQ
\|e_i\|  \leq  \left\| \frac{\partial S( s )}{\partial s} \right\|
\left\|
\left[
\begin{array}{cc}
|\widetilde \psi_i| |w_i^r| \\
|\widetilde \psi_i||v_i^r|
\end{array}
\right]
\right\|
 + \left\| S(\psi)\right\|
|\widetilde w_i | \\
 \leq  \left\| \frac{\partial S( s )}{\partial s} \right\|
\left\|
\left[
\begin{array}{cc}
| \varphi_i |  \\
|\widetilde \psi_i||v_i^r|
\end{array}
\right]
\right\|
 + \left\| S(\psi)\right\|
|\widetilde w_i |
\ENN
for some $s$ between $\psi_i^r$ and $\psi_i^r+\widetilde\psi_i$, noting $\tilde v_i=0$.
Since the norms of the rotation matrix $S$ and its derivative, i.e.,
$\left\| S(\psi)\right\|$ and $\left\| {\partial S( s )}/{\partial s} \right\| $,
are always bounded,  $|v_i^r|$ is bounded, and
$\lim_{t\rightarrow\infty}  | \varphi_i(t)| =0$,
$\lim_{t\rightarrow\infty}  | \tilde\psi_i(t)| =0$,
$\lim_{t\rightarrow\infty}  | \tilde\omega_i(t)| =0$,
one has
$ \lim_{t\rightarrow\infty}   e_i(t) =0$, exponentially.} The proof is thus completed.
\end{proof}

\begin{remark}
When the vessels work in a low frequency motion scenario with
the desired state trajectories $w_i^r(t)$ and $\psi_i^r(t)$ varying slowly
and $w_i^r(t)$ bounded by a constant $\varpi_i= \varpi_o$, one can simplify the controllers
by approximately using
$\dot w_i^r=0,\ddot w_i^r=0, \dot \varpi_i=0,\ddot \varpi_i=0, \dot \psi_i^r=0, \ddot \psi_i^r =0$.
Then, the controller (\ref{tau1}) reduces to
 \EQ
 \barray{rll}
\tau_{i,1} &=&- (k_1/ k_3 )w_i - (k_2/ k_3)v_ir_i   \\
&& -(\kappa_1/ k_3)  \int  \widetilde w_i(s) ds   -(\kappa_2/ k_3) \widetilde w_i , \label{PID1}
\earray
\EN
and the controller (\ref{tau2}) to
  \EQ
\tau_{i,2} =
- \frac{ k_4       + \kappa_3
+\kappa_4   }{ k_5  } \dot\psi_i  + \frac{
 \kappa_3^2     - (\kappa_3
+\kappa_4) \kappa_3   }{k_5 }\widetilde \psi_i . \label{PID2}
\label{eq: u2}\EN
{\blue Obviously, the simplified controller $\tau_{i,1}$ is of the Proportional-Integral (PI)  form with proper feedforward compensation
and $\tau_{i,2}$ is of the Proportional-Derivative (PD) form.}
\end{remark}

\section{Experiments}\label{se: experiments}

The multi-USV experimental platform  is composed of three HUST-12 vessels,  a target vessel, a differential NovAtel-OEM615-typed GPS station, and a wireless LAN hub as shown in Fig.~{\ref{fig: multiusvs}}.
A HUST-12 USV  has a  plastic hull, an onboard differential GPS  receiver, a wireless sensor, an embedded controller,  a motor driver, a waterjet propellor, and a rudder. In particular, the waterjet propellor is composed of a propellor body, a reversing bucket, a jetting nozzle, and a motor.  The USV parameters are listed in Table~\ref{lab: usvparameters}.

\begin{figure}[htbp]
\centering
\includegraphics[width=8.5cm]{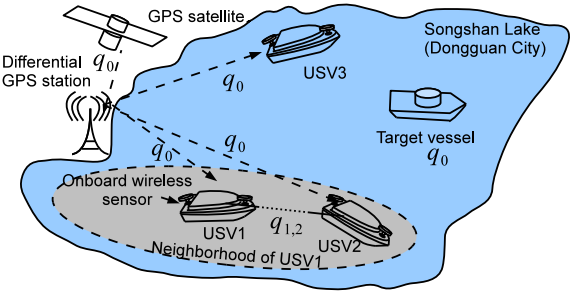}
\caption{Illustration of the multi-USV chasing system.}
\label{fig: multiusvs}
\end{figure}

%


%

\begin{table}
\caption{The physical characteristics parameters of HUST-12}
\begin{tabular*}{9cm}{llll}
\hline
Parameters & Value\\
\hline
Overall length  &  1.40m \\
Length on the waterline & 1.26m \\
Overall beam & 0.45m \\
Beam on the waterline & 0.39m  \\
Depth in water & 0.1m \\
Area of the waterplane & $0.44m^2$ \\
Longitudinal center of gravity to bow & 0.8m \\
Mass & 25kg \\
Maximal forward speed & 10m/s  \\
Maximal astern speed & 0.3m/s \\
Minimal turning radius & 1.1m \\
Cornering speed & 65deg/s \\
\hline
\label{lab: usvparameters}
\end{tabular*}\vspace{-5mm}
\end{table}

The waterjet propeller is actuated by the thrust of the motor, which is regulated by its rotating speed
of the range $[600,  11000]$ RPM. The engine is a brushless DC motor driven by a 44.4V electrical battery. The angle of the jetting nozzle within $[-20^\circ,20^\circ]$ is controlled by the steering engine.


The navigation system  consists of a differential GPS navigator, an accelerometer, and a gyro. Specifically, the differential GPS navigator is composed of a receiving board and two antennas with 5Hz frequency bandwidth. {\blue It
detects the velocities and positions of all the USVs as well as the target vessel with positioning accuracy  $\pm2$cm}. The merits of a differential GPS lie in its zero offset and low measuring noise. Meanwhile, the accelerometer and gyro are used to improve the gesture detection response of the GPS navigator. {\blue
Each USV can obtain the information of its neighbors and/or the target
by 433M wireless communication within the range of $125$m.
The control algorithm is executed by STM32F4 series CPUs of STMicroelectronics.}


{\blue First, the coefficients of the model~(\ref{dmodel}) are identified by zigzag tracking experiments, which are
{\red
\EQQ
k_1=-0.098, k_2=0.003, k_3=0.005,
k_4=-0.1055,  \\ k_5=0.019, k_6 = -0.091, k_7 = -0.0175.
\ENN } The data were collected on the present platform shown in Fig.~\ref{fig: multiusvs} with the sampling period of 0.2~s, which is compatible to the differential GPS updating frequency. }The surge, sway and yaw speed identification performances of the model~(\ref{dmodel}) are shown in Fig.~\ref{fig: identimodel}, with root-mean-square errors less than 5\%. The feasibility of the identified model~(\ref{dmodel}) is thus verified.

\begin{figure}[htbp]
\centering
\includegraphics[width=7cm]{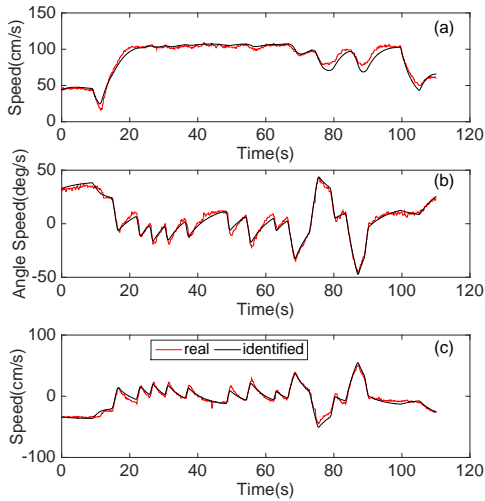}
\caption{Model identification for surge velocity $w_i$ (a),
sway velocity $v_i$ (b), and yaw angular velocity $r_i$ (c).}
\label{fig: identimodel}
\end{figure}

The PID controllers~(\ref{PID1}) and (\ref{eq: u2}) were first tested for a single vessel with the parameters $\kappa_1=0.02$, $\kappa_2=0.001$, $\kappa_3=0.076$, and $\kappa_4=0.418$. The corresponding control performance in terms of orientation and speed is shown in Fig.~\ref{fig: pidangle}. Therein,  the set points are $\omega_i^r= 200$cm/s and $\psi_i^r=300^{\circ}$, respectively. The overshoots/settling times ($\pm 0.5\%$ threshold) of the speed and orientation are $2\%$/10s and $1\%$/15s, respectively, which satisfies the technical requirements of the collective surrounding control.

\begin{figure}[htbp]
\centering
\includegraphics[width=9cm]{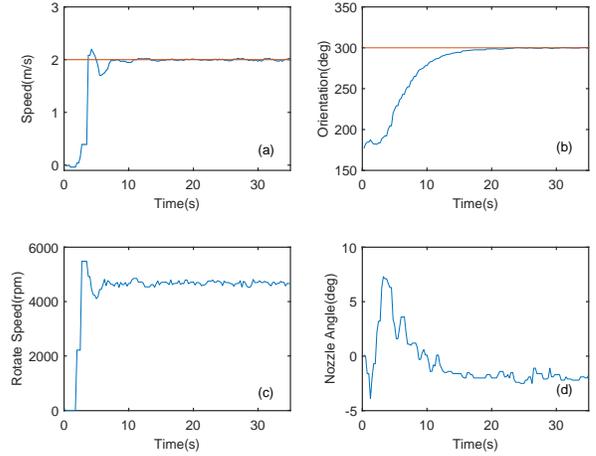}
\caption{The control performances of the single USV speed (subfigure~(a)) and orientation (subfigure~(b)) with the simplified control law  (\ref{PID1})  and (\ref{eq: u2}). The profiles of the control signals, i.e., the propellor rotational speed $\tau_{i,1}$  and the nozzle angle $\tau_{i,2}$ are plotted in subfigures~(c) and~(d), respectively. The red straight lines represent the set-points $\omega_i^r$ and $\phi_i^r$.}
\label{fig: pidangle}
\end{figure}

Upon the lower level controllers, we also conducted collective surround control for the multi-USV system.
The results are shown in Fig.~{\ref{fig: boatrealprocess}}.
%
%
Initially,  all the USVs were randomly distributed on a $[40\times40]$~m$^2$ water surface area and began to chase the target vessel.  Within 50 seconds, the USVs caught up the target and began to collectively move around it, using the asymptotically equally surrounding control law~(\ref{controleta}) and (\ref{eq: controlomega})
with the parameters  $\beta_1=0.13$, $\beta_2=0.06$, $\rho_0=10$~m.
In the 160th second, the equally surrounding control mission was fulfilled, that is,  the multi-USVs  captured the target and collectively surround it with fixed distances and evenly-distributed phases.

The collective surrounding procedure is also shown in Fig.~\ref{fig: realline}(a) and (b), in terms of
the USV-target and inter-USV distances,  respectively. It is observed that the three USV-target distances asymptotically converges to $\rho_0$, and the three inter-USV phases asymptotically converge to $120^\circ$. The feasibility of the proposed two-level controller, i.e., the upper level equally surrounding controller~(\ref{controleta}) and (\ref{eq: controlomega}) and  the lower level single vessel controller~(\ref{tau1}) and (\ref{tau2}), is thus verified.

\begin{figure}[htbp]
\centering
\includegraphics[width=9cm]{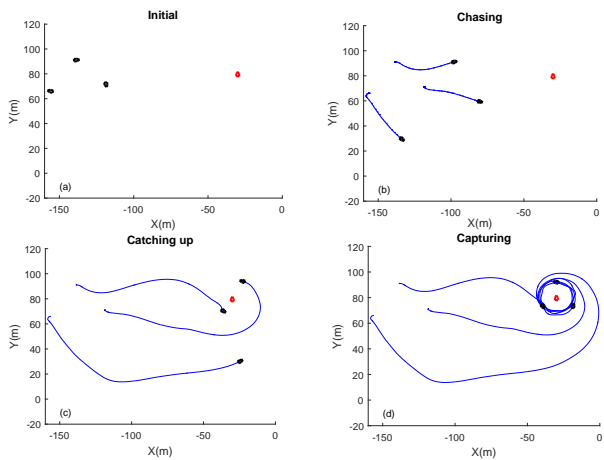}
\caption{The collective surrounding control procedures of the multi-USV system, where the black and red points represent the USVs and the target vessel, respectively; and the blue and red lines represent the moving trajectories of the USVs and the target vessel, respectively.  The procedures include
(a): initial positions;  (b): the target vessel is chased by USVs;  (c): the target vessel is caught up; (d): the target vessel is captured and surrounded.}
\label{fig: boatrealprocess}
\end{figure}


\begin{figure}[htbp]
\centering
\includegraphics[width=8cm]{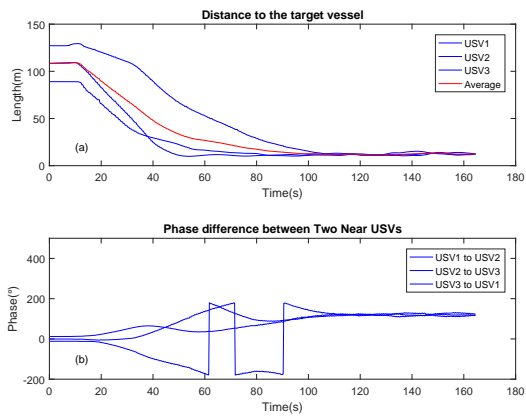}
\caption{ Evolutions of (a) USV-target distances and (b) inter-USV phase. The sudden turns of the curves are due to the phase range of $[-180^\circ,180^\circ)$. Finally, the three inter-USV phases converge to $120^\circ$.  }
\label{fig: realline}
\end{figure}

\section{Conclusions}\label{se: conclusion}
In this paper, a two-level onboard distributed surrounding controller has been developed, which has also been tested on a real multi-USV experimental platform.
 The research has verified the success of the developed hardware and software architectures. {\blue The proposed theoretical framework does not accommodate external disturbance, which is sufficiently satisfactory in the current scenario, as assessed by experiments. However, a more complete theoretical
framework for external disturbance rejection has to be considered in severer environments, which will be the future work.
The current platform was developed in the medium speed mode (1-3m/s) while the issues for high speed mode (8-10m/s) will also be the future work, such as
significant communication package loss.}

\bibliographystyle{ieeetr}
\bibliography{IEEEfull,aaa}
%
%


%



\end{document}